\documentclass[11pt]{article}
\usepackage{amssymb,amsmath,amsthm}
\usepackage{epsfig}
\usepackage{url}
\PassOptionsToPackage{hyphens}{url}\usepackage{hyperref} %because arXiv broke the url in the foonote without this

\makeatletter
\let\@@citation@@=\citation
\renewcommand{\citation}[1]{\@@citation@@{#1}%
\@for\@tempa:=#1\do{\@ifundefined{cit@\@tempa}%
  {\global\@namedef{cit@\@tempa}{}}{}}%
}
\makeatother

\usepackage{hyperref}
\hypersetup{pdfpagemode=UseNone} %to avoid bookmarks opening: https://tex.stackexchange.com/questions/50814
\makeatletter
\def\@lbibitem[#1]#2#3\par{%
  \@ifundefined{cit@#2}{}{\@skiphyperreftrue
  \H@item[%
    \ifx\Hy@raisedlink\@empty
      \hyper@anchorstart{cite.#2\@extra@b@citeb}%
        \@BIBLABEL{#1}%
      \hyper@anchorend
    \else
      \Hy@raisedlink{%
        \hyper@anchorstart{cite.#2\@extra@b@citeb}\hyper@anchorend
      }%
      \@BIBLABEL{#1}%
    \fi
    \hfill
  ]%
  \@skiphyperreffalse}%
  \if@filesw
    \begingroup
      \let\protect\noexpand
      \immediate\write\@auxout{%
        \string\bibcite{#2}{#1}%
      }%
    \endgroup
  \fi
  \ignorespaces
  \@ifundefined{cit@#2}{}{#3}}

\def\@bibitem#1#2\par{%
  \@ifundefined{cit@#1}{}{\@skiphyperreftrue\H@item\@skiphyperreffalse
  \Hy@raisedlink{%
    \hyper@anchorstart{cite.#1\@extra@b@citeb}\relax\hyper@anchorend
    }}%
  \if@filesw
    \begingroup
      \let\protect\noexpand
      \immediate\write\@auxout{%
        \string\bibcite{#1}{\the\value{\@listctr}}%
      }%
    \endgroup
  \fi
  \ignorespaces
  \@ifundefined{cit@#1}{}{#2}}
\makeatother

\newtheorem{thm}{Theorem}%[section]
\newtheorem{cor}[thm]{Corollary}
\newtheorem{lem}[thm]{Lemma}

\usepackage{indentfirst}
\usepackage{xspace} % A makrok utani space korrekt kezelese miatt.

\def\N{\mbox{\ensuremath{\mathbb N}}\xspace}

\def\Q{\mbox{\ensuremath{\mathbb Q}}\xspace}
\def\Z{\mbox{\ensuremath{\mathbb Z}}\xspace}

\mathcode`l="8000
\begingroup
\makeatletter
\lccode`\~=`\l
\DeclareMathSymbol{\lsb@l}{\mathalpha}{letters}{`l}
\lowercase{\gdef~{\ifnum\the\mathgroup=\m@ne \ell \else \lsb@l \fi}}%
\endgroup

\usepackage{authblk}

\begin{document}

\title{The range of non-linear natural polynomials cannot be context-free}
\author{D\"om\"ot\"or P\'alv\"olgyi
%Research supported by the Lend\"ulet program of the Hungarian Academy of Sciences (MTA), under grant number LP2017-19/2017
%Biro Petis: XY was supported by the Hungarian Academy of Sciences under the research grant no. KEP-6/2017
%Hungarian National Science Fund (OTKA), under grant PD 104386 and NN 102029 (EUROGIGA project GraDR 10-EuroGIGA-OP-003), the J\'anos Bolyai Research Scholarship of the Hungarian Academy of Sciences, and by the Marie Sk\l odowska-Curie action of the European Commission, under grant IF 660400.
}
\affil{MTA-ELTE Lend\"ulet Combinatorial Geometry Research Group,\\
	Institute of Mathematics, E\"otv\"os Lor\'and University, Budapest, Hungary}
%\date{}
\maketitle

\begin{abstract}
	Suppose that some polynomial $f$ with rational coefficients takes only natural values at natural numbers, i.e., $L=\{f(n)\mid n\in \N\}\subseteq\N$.
	We show that the base-$q$ representation of $L$ is a context-free language if and only if $f$ is linear, answering a question of Shallit.
	The proof is based on a new criterion for context-freeness, which is a combination of the Interchange lemma and a generalization of the Pumping lemma.
\end{abstract}

\medskip

Call a polynomial $f$ over \Q \emph{natural} if $f(n)\in \N$ for every $n\in \N$.
For example, $\frac{x^2+x}2$ is natural.
Shallit \cite[Reseach problem 3 in Section 4.11, page 138]{Shallit} proposed to study whether the base-$q$ representation of the range, $L=\{f(n)\mid n\in \N\}$, of a natural polynomial is context-free or not.
It is easy to see that if $f$ is linear, i.e., its degree is at most one, then $L$ is context-free for any $q$.
It was conjectured that $L$ is not context-free for any other $f$.
This conjecture was known to hold only in special cases, though S\'andor Horv\'ath had an unpublished manuscript that claimed a solution.\footnote{According to Shallit, see \url{https://cstheory.stackexchange.com/a/41864/419.}.}
The goal of this note is to present a simple proof that uses a new lemma, which is a simple combination of two well-known necessary criteria for the context-freeness of a language.\\

A \emph{context-free grammar} $G$ is defined as a finite $4$-tuple $G=(V,\Sigma,P,S)$, where $V$ is the set of \emph{non-terminal symbols}, $\Sigma$ is the set of the \emph{terminal symbols}, which we also call the letters of the \emph{alphabet} (where $V\cap \Sigma=\emptyset$), $P$ is the set of \emph{production rules} and $S\in V$ is the \emph{start symbol}.
Each production rule is of the form $A\to \alpha$ where $A\in V$ and $\alpha \in (V\cup \Sigma)^*$ is a \emph{string}.
When such a rule is applied to an occurrence of $A$ in some string $\beta$, that occurrence of the symbol $A$ is replaced with $\alpha$ in $\beta$ to obtain a new string.
We say that a string $\gamma\in (V\cup \Sigma)^*$ can be \emph{derived} from another string $\beta \in (V\cup \Sigma)^*$ if after applying some rules to certain occurrences of the appropriate non-terminals starting from $\beta$ we can obtain $\gamma$.
The \emph{language} $L(G)$ of the grammar $G$ is the set of words from $\Sigma^*$ that can be derived from $S$.
A derivation of a word $z\in L(G)$ from $S$ can be described by a \emph{derivation tree}; this is a rooted ordered tree whose non-leaf nodes are labeled with non-terminal symbols such that the root is labeled with $S$, the labels of the children of any node labeled $A$ are the right side of some rule $A\to \alpha$ in the given order, and leaves are labeled with terminal symbols that give $z$ in the given order. 
A grammar is in \emph{Chomsky normal form} if the right side of each production rule is either two non-terminal symbols, or one terminal symbol, or the empty string; every context-free grammar has a Chomsky normal form.
A language $L$ is context-free if $L=L(G)$ for some context-free grammar $G$.
For other basic definitions and statements about context-free grammars and languages, we direct the reader to \cite{Shallit}.\\

Now we state two lemmas that we later combine.\footnote{Note that we here we state them in a slightly weaker form as their original versions, as we do not use some parts of the original statements.}
The first is known as the Interchange Lemma.

\begin{lem}[Interchange Lemma \cite{ORW82}]\label{interchange}
	For every context-free language $L$ there is a constant $p>0$ such that for all $n\in \N$ for any collection of length $n$ words $R\subset L$ there is a subset $Z=\{z_1,\ldots,z_k\}\subseteq R$ with $k\ge |R|/(pn^2)$, and decompositions $z_i=v_iw_ix_i$ such that each of $|v_i|$, $|w_i|$, and $|x_i|$ is independent of $i$, and the words $v_iw_jx_i$ are in $L$ for every $1\le i,j \le k$.
\end{lem}

The second is the following generalization of the Pumping Lemma \cite{BH61}.

\begin{lem}[D\"om\"osi-Kudlek \cite{DK99}]\label{O1G}
	For every context-free language $L$ there is a constant $p$ such that if in a word $z\in L$ we distinguish $d$ positions and exclude $e$ positions such that $d\ge p(e+1)$, then there is a decomposition $z=uvwxy$ such that
	 $vx$ has a distinguished position, but no excluded positions and  $uv^iwx^iy\in L$ for every $i\ge 0$.
\end{lem}

A straight-forward combination of the proofs of Lemmas \ref{interchange} and \ref{O1G} gives the following.

\begin{lem}[Combined lemma]\label{combined}
	For every context-free language $L$ there is a constant $p>0$ such that for all $n$ for any collection of length $n$ words $R\subseteq L$, if we distinguish $d$ positions and exclude $e$ positions such that $d\ge p(e+1)$, then there is a $Z=\{z_1,\ldots,z_k\}\subset R$ with $k\ge |R|/(pn^4)$, and a decomposition $z_i=u_iv_iw_ix_iy_i$ for every $1\le i\le k$ such that
		\begin{itemize}
		\item $|u_i|$, $|v_i|$, $|w_i|$, $|x_i|$, and $|y_i|$ are all independent of $i$,	
		\item $v_ix_i$ has a distinguished position, but no excluded positions,
		%\item the number of distinguished positions in $v_iw_ix_i$ is at most $p(e+1)$,
		\item $u_{i_0}v_{i_1}\ldots v_{i_m}w_{i_{m+1}}x_{i_m}\ldots x_{i_1}y_{i_0}\in L$ for every sequence of indices $1\le i_0,i_1,\ldots,i_{m+1} \le k$.
	\end{itemize}
\end{lem}

The proof of Lemma \ref{combined} %is a simple combination of the proofs of Lemmas \ref{interchange} and \ref{O1G}, and it
can be found at the end of this note.
Now we state an interesting corollary of Lemma \ref{combined} that we can apply to Shallit's problem.

\begin{cor}\label{finite}
	If in a context-free language $L$ for infinitely many $n$ there are $\omega(n^4)$ words of equal length in $L$ whose first $\omega(n)$ letters are the same and their last $n$ letters are different (pairwise), then there is an integer $B$ such that there are infinitely many pairs of words in $L$ of equal length that differ only in their last $B$ letters.
\end{cor}
\begin{proof}
	There is a $p$ that satisfies the conditions of Lemma \ref{combined} for $L$. 
	Take a large enough $n$ for which there are $pn^4+1$ words of equal length in $L$ whose first $p(n+1)$ letters are the same, but their last $n$ letters are different; this will be $R$.
	Apply Lemma \ref{combined} to $R$, distinguishing the first $p(n+1)$ positions and excluding the last $n$ positions to obtain some $Z=\{z_1=u_1v_1w_1x_1y_1,z_2=u_2v_2w_2x_2y_2\}$.
	It follows from the conditions that $u_1$ and $u_2$ must contain only distinguished positions, thus $u_1=u_2$.
	Since $v_i$ and $x_i$ cannot contain excluded positions, either $y_1\ne y_2$, or $|x_1|=|x_2|=0$ and $w_1y_1\ne w_2y_2$.
	In the former case the pairs of words $u_1v_1^jw_1x_1^jy_1$ and $u_2v_1^jw_1x_1^jy_2=u_1v_1^jw_1x_1^jy_2$, in the latter case the pairs of words $u_1v_1^jw_1y_1$ and $u_1v_1^jw_2y_2$ satisfy the conclusion for $j>0$.
\end{proof}

Now we are ready to prove our main result.

\begin{thm}
	$L$ is not context-free for non-linear natural polynomials over any base-$q$.
\end{thm}
\begin{proof}
	First we show that the condition of Corollary \ref{finite} is satisfied for every natural polynomial $f$ for infinitely many $n$ for some words from $L=\{f(x)\mid x\in \N\}$. %Set $t=p-1$ and $n=C^p$ for some large enough $C=C(f)$.
	The plan is to take some numbers $x_1,\ldots,x_{N}$ (where $N=n^5$) for which $f(x_i)\ne f(x_j)$, and then add some large number $s$ to each of them to obtain the desired words $f(x_i+s)$.
	
	If the degree of $f$ is $d$, then at most $d$ numbers can take the same value, thus we can select $x_1,\ldots,x_{N}$ from the first $dN$ numbers, which means that they have $O(\log n)$ digits (since $d$ is a constant). In this case $f(x_i)=O((dN)^d)$, thus each $f(x_i)$ will also have $O(\log n)$ digits. If we pick $s$ to be some number with $n^2$ digits, then $f(s)$ will have $D=dn^2+\Theta(1)$ digits, and each $f(x_i+s)$ will have $D$ or $D+1$ digits, thus at least half, i.e., $N/2$ of them have the same length; these will be the words we input to Corollary \ref{finite}.	
	We still need to show that for these $f(x_i+s)$ their first $\Omega(n^2)$ digits are the same and that their last $O(\log n)$ digits differ.
	
	 Let $M\in \N$ be such that $f(x)=\sum_{i=0}^d \frac{\alpha_i}M x^i$ for $\alpha_i\in \Z$. If $s$ is a multiple of $Mq^m$, then the last $m$ digits of $f(x)$ and $f(x+s)$ are the same for any $x$. This way it is easy to ensure that the last $O(\log n)$ digits in base-$q$ stay different.
	 Since $f(x+s)=\sum_{i=0}^d \frac{\alpha_i}M (x+s)^i=\frac{\alpha_d}M s^d+O(s^{d-1}(dN)^d)$, the first $n^2-O(\log n)$ digits can take only two possible values (depending on whether there is a carry or not), thus one of these values is the same for $N/2$ of the $f(x_i+s)$.
	 Thus we have shown that the condition of Corollary \ref{finite} is satisfied

	If $L=\{f(x)\mid x\in \N\}$ was context-free, then from the conclusion of Corollary \ref{finite} we would obtain infinitely many pairs of numbers, $a_i,b_i \in L$, such that $|a_i-b_i|\le 2^B$, but this is impossible for non-linear polynomials.
\end{proof}

We end with the omitted proof.

\begin{proof}[Proof of Lemma \ref{combined}]
	Fix a context-free grammar for $L$ in Chomsky normal form, with $t$ non-terminals.
	Fix a derivation tree for each word $z\in R$.
	We say that a node has a distinguished (resp.\ excluded) \emph{descendant} if a distinguished (resp.\ excluded) position is derived from the given node in the tree, i.e., if there is a leaf among its descendants whose label is in a distinguished (resp.\ excluded) position of $z$.
		
	Call a node of the derivation tree an \emph{e-branch} node if both of its children have an excluded descendant.
	There are exactly $e-1$ e-branch nodes in the derivation tree (if $e\ge 1$).
	
	Call a node of the derivation tree a \emph{d-branch} node if both of its children have a distinguished descendant.
	There are exactly $d-1$ d-branch nodes in the derivation tree.
	Say that a d-branch node is the \emph{d-parent} of its descendant d-branch node if there are no d-branch nodes between them.
	With this structure the d-branch nodes form a binary tree.
	The $i^{\mathrm{th}}$ d-parent of a d-branch node is the $i$-times iteration of the d-parent operator.
	
	Call a d-branch node \emph{bad} if there is an e-branch node between it and its $(2t+3)^{\mathrm{th}}$ d-parent (excluding the node, but including its $(2t+3)^{\mathrm{th}}$ d-parent), or if it does not have a $(2t+3)^{\mathrm{th}}$ d-parent.
	Because of the binary structure of the d-branch nodes, each e-branch node can cause at most $2^{2t+3}$ d-branch nodes to be bad, and a further $2^{2t+3}$ d-branch nodes might not have a $(2t+3)^{\mathrm{th}}$ d-parent.
	Therefore in total there are at most $e2^{2t+3}$ bad d-branch nodes, so there is a d-branch node that is not bad if $p>2^{2t+3}$.
	Consider the path from the $(2t+3)^{\mathrm{th}}$ d-parent of a non-bad d-branch node to the non-bad d-branch node.
	Note that the nodes on this path might have an excluded descendant, but since there is no e-branching node along the path, we can conclude that there is a subpath with $t+1$ d-branch nodes on it such that no sibling of any node along the subpath has an excluded descendant.
	(The worst case is when the excluded descendant(s) belong to the sibling of the $(t+2)^{\mathrm{nd}}$ d-parent of a non-bad d-branch node.)
	
	By the pigeonhole principle some non-terminal $A$ appears twice on the left side of a rule along this subpath.
	While we reach one node from the other, some string $\alpha A\beta$ is derived from $A$.
	Apply the corresponding rules from the derivation tree to $\alpha$ and $\beta$ to obtain the string $vAx$ where $v,x\in\Sigma^*$.
	Thus, % we have $A \stackrel *\Rightarrow vAx$
	$z$ can be written as $z=uvwxy$ such that $vAx$ can be derived from $A$, $w$ can be derived from $A$, the subwords $v$ and $x$ have no excluded position (since they are descendants of siblings of nodes along the path), but at least one of them has a distinguished position.
	For each $z\in R$ we fix such a decomposition $z=uvwxy$.
	
	We partition $R$ into at most $t\binom{n+4}4$ groups depending on which non-terminal $A$ appeared on the left side of the rule, and the lengths of $u,v,w,x$ and $y$.
	Let $c=t \max_n \binom{n+4}4/n^4=5t$ (if we only care about large $n$, then $c$ would be close to $t/24$).
	By the pigeonhole principle one of the groups will have at least $|R|/(cn^4)$ words in it; this will be $Z$.
	Since we can arbitrarily apply the rules for $A$, the conclusion follows with $p\ge\max(c,2^{2t+3}+1)$.
\end{proof}

\subsubsection*{Remark} %s and acknowledgment}
This note started as a CSTheory.SE answer.\footnote{See my two answers for \url{https://cstheory.stackexchange.com/questions/41863/base-k-representations-of-the-co-domain-of-a-polynomial-is-it-context-free}; note that at that time I didn't know about Lemmas \ref{interchange} and \ref{O1G}.}

\end{document}